\documentclass[a4paper,UKenglish]{article}

\usepackage{a4wide}
\usepackage{microtype}
\usepackage{xspace}

\title{Fast Sketch-based Recovery of Correlation Outliers\footnote{The
    work of GC is supported by European Research Council grant ERC-2014-CoG 647557 and a Royal Society Wolfson Research Merit Award; JD is supported by a Microsoft Research PhD Scholarship
    (MRL 2014-038).}}

\author{
  Graham Cormode\\
  \texttt{g.cormode@warwick.ac.uk}
  \and
  Jacques Dark\\
  \texttt{j.dark@warwick.ac.uk}
  }

\usepackage[ruled,linesnumbered,noend]{algorithm2e}
\usepackage{amsmath}
\usepackage{amssymb}
\usepackage{amsthm}

\theoremstyle{plain}
\newtheorem{theorem}{Theorem}
\newtheorem{proposition}[theorem]{Lemma}
\newtheorem{corollary}[theorem]{Corollary}
\theoremstyle{definition}
\newtheorem{definition}[theorem]{Definition}

\newenvironment{informal}
  {\begin{proof}[Informal Proof]}
  {\end{proof}}

\newcommand{\E}[1]{\mathbb{E}\left[#1\right]}
\newcommand{\Prob}[1]{\mathbb{P}\left[#1\right]}
\newcommand{\Var}[1]{\text{Var}\left[#1\right]}

\newcommand{\RPS}{\textsf{RPS}\xspace}
\newcommand{\CPS}{\textsf{CPS}\xspace}
\newcommand{\TS}{\textsf{TS}\xspace}

\newcommand{\para}[1]{\medskip \noindent {\bf #1}}

\begin{document}

\maketitle

\begin{abstract}\label{sec:abstract}
Many data sources can be interpreted as time-series, and a key problem is to identify which pairs out of a large collection of signals are highly correlated. We expect that there will be few, large, interesting correlations, while most signal pairs do not have any strong correlation. We abstract this as the problem of identifying the highly correlated pairs in a collection of $n$ mostly pairwise uncorrelated random variables, where observations of the variables arrives as a stream. Dimensionality reduction can remove dependence on the number of observations, but further techniques are required to tame the quadratic (in $n$) cost of a search through all possible pairs. 
  
We develop a new algorithm for rapidly finding large correlations based on sketch techniques with an added twist: we quickly generate sketches of random combinations of signals, and use these in concert with ideas from coding theory to decode the identity of correlated pairs. We prove correctness and compare performance and effectiveness with the best LSH (locality sensitive hashing) based approach.
\end{abstract}

\section{Introduction}\label{sec:intro}
One of the most basic tasks in data analysis is to identify correlations between data sources, modeled as random variables. Discovered correlations are used to remove unnecessary features, to build predictive models,
and to identify unexpected behaviors and dependencies. In this paper, we consider the most common measure of correlation: the Pearson product-moment correlation coefficient, which describes the linear relationship between a pair of random variables. This measure is simple to state and interpret: it is computed as the (sample) covariance of the two variables, divided by the product of the corresponding standard deviations. It ranges from $-1$ (strong negative correlation) through $0$ (no correlation) to $+1$ (strong positive correlation). Hence, we are typically interested only in attribute pairs with correlation close to $1$ in magnitude.

For large numbers of variables, it can quickly become infeasible to
compute the correlations of all of the quadratically many
pairs. However, our observation is that most correlations are
uninteresting: for many kinds of data, we expect that most pairs of
variables would \emph{not} display any (strong) correlation.
For example, if we consider the activity profiles of users of a large
web service, then we do not expect many pairs to be strongly
correlated (there may be weak correlations due to similar time-of-day
and day-of-week behavior) - any strong correlation between a pair
would be unusual, indicating potentially
nefarious activity worthy of further investigation.
We model this by assuming that the number of correlated pairs is
asymptotically smaller than the quadratically many possible pairs. 

With this in mind, we can ask the following questions: given a stream of observation data, can we identify all correlation outliers (unusually large correlation coefficients, defined by being greater in magnitude than some parameter $\phi$) with query time cost sub-quadratic in the number of variables and sub-linear in the number of observations?

This can be done using a combination of a Fast Johnson-Lindenstrauss Transform (FJLT, to compress the rows of the input matrix) and Locality Sensitive Hashing (LSH, to efficiently find the outlier pairs). However, for small $\phi$, the query time of this strategy looks like $n^{2-\Theta(\phi)}$, even as we shrink all the non-outlier correlations down to $0$. Can we improve the query time in this case?
This paper answers this question in the positive, by describing an algorithm which takes sketches of the rows and uses fast matrix multiplication to quickly transform them into an approximation of a sketch of the correlation matrix. We then remove the $1$'s along the diagonal, and use a heavy hitters recovery technique to pull out the outliers. For fixed Frobenius
norm\footnote{For matrix $\mathbf{A}$, the Frobenius norm is
  $\textstyle \lVert \mathbf{A} \rVert_F = (\sum\limits_{i,j}\mathbf{A}_{i,j}^2)^\frac{1}{2}$.} of the non-outlier non-diagonal correlations, this query process can be performed in time $\tilde{O}(\phi^{-2}n^{5/3})$, assymptotically better than LSH for $\phi < 1/3$. However, this comes at the cost of requiring much larger sketches of the input matrix rows.

\section{Preliminaries}\label{sec:prelims}
\subsection{Models}\label{sec:models}
We treat the observation data as defining an $n \times p$ matrix of
reals, $\mathbf{M}$.
Here, $n$ denotes the number of attributes, while $p$ indexes the
different observations.
Hence, each of the $p$ columns represents an independent observation
of some $n$-dimensional random variable.
We label the columns (observations) as $\mathbf{x^{(i)}}$ for $i \in [p]$.
For this data, we can apply standard definitions of covariance and
correlation.

\begin{definition}\label{def:estimators}
Recall that:
\begin{itemize}
\item
  The \emph{sample mean} is given by $\mathbf{\overline{x}} = \frac{1}{p} \sum_{i=1}^p \mathbf{x^{(i)}}$.
\item
  The \emph{sample covariance} is given by
  $\mathbf{V} = \frac{1}{p-1}(\mathbf{M} - \mathbf{\overline{x}} \mathbf{e}^T) (\mathbf{M} - \mathbf{\overline{x}} \mathbf{e}^T)^T \text{,}$
  where $\mathbf{e}$ is the $p$-dimensional vector with entries all ones.
\item
  The \emph{sample correlation} is given by $\mathbf{C} = \mathbf{\Sigma}^{-\frac{1}{2}} \mathbf{V} \mathbf{\Sigma}^{-\frac{1}{2}}$, where $\mathbf{\Sigma}$ is the diagonal matrix consisting of the diagonal entries of $\mathbf{V}$.
\end{itemize}
\end{definition}

Essentially, the covariance is found by shifting the rows of
$\mathbf{M}$ to have mean $0$ and then taking inner products between
them normalized by a factor of $\frac{1}{p-1}$.
The definition of the correlations is similar, but further normalized by diving out the standard deviations.

It will be useful for our analysis to have notations for the rows of $\mathbf{M}$ with the shift normalization applied.

\begin{definition}\label{def:standardized}
  For each row vector $\mathbf{y^{(i)}}$:
  \begin{itemize}
  \item
    Let the standardized row vector $\mathbf{\hat{y}^{(i)}}$ be given by $\mathbf{\hat{y}^{(i)}} = \frac{\mathbf{\hat{y}^{(i)}} - \mathbf{\bar{x}_ie}^T}{\| \mathbf{\hat  {y}^{(i)}} - \mathbf{\bar{x}_ie}^T\|_2}$.
  \end{itemize}
\end{definition}

The observation matrix $\mathbf{M}$ is input as a stream of $m$
updates $\langle u_1, u_2, \cdots u_m \rangle$ arriving one at a
time.
Starting from the zero matrix $\mathbf{M^{(0)}} = \mathbf{0}$, each update $u_s$ describes a change to be made to $\mathbf{M^{(s-1)}}$ in order to determine $\mathbf{M^{(s)}}$. By the end of the stream, we have $\mathbf{M^{(m)}} = \mathbf{M}$. The format of the updates depends on the exact choice of stream model---we will consider three variants: row-wise permutation, column-wise permutation, and turnstile.

\para{Row-Wise Permutation Stream (\RPS).}
In this model, the updates are simply a list of the entries of
$\mathbf{M}$, one row at a time. With each step from $\mathbf{M^{(s-1)}}$ to
$\mathbf{M^{(s)}}$, one entry is changed from $0$ to its final
value.
Entries in the same row arrive contiguously, so each row is
filled out one after the other.
Without loss of generality, we can assume that rows arrive in index
order, so that
$\mathbf{M}_{i,j} \gets u_{ip+j} $
%
%The order of the arrival of the rows
%themselves, and the order of arrival of entries within rows are
%arbitrary.
Since each entry is set exactly once, the stream has length $m =
np$.
The arrival of each new row corresponds to adding a new
attribute to the data set.

\para{Column-Wise Permutation Stream (\CPS).}
This model works the same as the row-wise version, but with entries
arriving as contiguous columns.
Again, $m = np$ but now
$\mathbf{M}_{i,j} \gets u_{jn + i} $.
The arrival of a new column corresponds to adding a new observation
(e.g. from a new time step).

\para{Turnstile Stream (\TS).}
The turnstile model is the most general that we consider. Here updates
are of the form $u_t = (\alpha, i, j)$ indicating that the $(i,
j)$\textsuperscript{th} entry should be incremented by $\alpha \in
\mathbb{R}$.
That is, $\mathbf{M}^{(s)}_{i,j} \gets \mathbf{M}^{(s-1)}_{i,j} + \alpha$,
while all other entries remain the same.
Changes happen in any order, and entries can change any number of
times as long as the correct state is reached by the end of the
stream.
Hence, the stream length $m$ is arbitrary.

Both \RPS and \CPS are then special cases of this model. \TS
represents the situation where each of the observed values needs to be
aggregated from a variety of sources.
For example: suppose the entries in our observation matrix represent the number of requests for a specific resource (indexed by rows) at a specific site (indexed by columns) in a distributed system. Then, the number of requests at each node will need to be accumulated to produce the actual observation data.

\subsection{Problem Statement}\label{sec:problem}

In what we term the \emph{correlation outliers} problem, we are given
a stream describing $\mathbf{M}$ (according to one of the three
models), and three parameters: $k$, $\phi$, and $R$.
We make use of the following concepts:

\begin{definition}\label{def:large}
For the sample correlation matrix $\mathbf{C}$ (of $\mathbf{M}$):
\begin{itemize}
\item
  Let $\textsc{Large}_\phi \subset [n]^2$ refer to the set of index pairs of off-diagonal entries of $\mathbf{C}$ which have magnitude at least $\phi$.
\item
  Let $\mathbf{C}_{-k}$ refer to the matrix obtained by taking $\mathbf{C}$, removing all the diagonal entries, and removing the $k$ largest magnitude off-diagonal entries (replacing them with $0$'s).
\end{itemize}
\end{definition}

The problem is then to maintain a summay of the stream so that
all index pairs contained in $\textsc{Large}_\phi$ can be retrieved
with high probability ($o(\frac1n)$ chance of failure), provided
 that $|\textsc{Large}_\phi| \leq k$ and $\lVert \mathbf{C}_{-k}
 \rVert_F \leq R$.
Since the full input can be trivially maintained in $O(np)$ space, we
seek solutions with space cost that is $o(np)$.
Further, the summary should be quick to update (taking polylogarithmic time) and, at the end of the stream, the query routine should run in time $o(n^2)$. %and return a list of index pairs which contains all of $\textsc{Large}_\phi$ with high probability ($o(\frac{1}{n})$ chance of failure).
We argue that the assumption that $k$, the number of highly correlated
pairs, is $o(n^2)$ is a reasonable one: otherwise, simply reporting
all the correlated pairs would take quadratic time, and naive exhausitve
solutions would suffice.
A similar assumption is made in prior work considering Boolean vectors~\cite{BooleanCorrelation,FasterBooleanCorrelation}.

\subsection{Related Work}\label{sec:related}
% Since correlation coefficients are simply centered and normalized inner products (recall definition~\ref{def:estimators} - we centered to find covariance and then normalized for correlation), the correlation outliers problem can be related to two important problems: approximate matrix multiplication and approximate near neighbours on the unit Euclidean sphere.

% \begin{itemize}
%   \item
%   The matrix of all correlation coefficients can also be expressed as $\mathbf{\tilde{M}}\mathbf{\tilde{M}}^T$, where $\mathbf{\tilde{M}}$ is the matrix of centered and normalized rows of the observation matrix $\mathbf{M}$. This is a simple consequence of the fact that the normalization step commutes with the matrix product step in definition~\ref{def:estimators}. Hence, the correlation outliers problem can be seen as the problem of finding large entries in a matrix product, with the added complication of centering and normalizing.

%   \item
%   Since a large positive correlation is equivalent to small Euclidean distance between rows of $\mathbf{\tilde{M}}$, we can find large positive correlations by looking for pairs that are aNN (approximate near neighbours) on the unit Euclidean sphere. Adding a negated duplicate of each row also allows us to find large negative correlations. However, in aNN it is assumed you have full access to the input, whereas we wish to work with a stream.
% \end{itemize}

\para{Locality Sensitive Hashing.}
Asking for high correlation is equivalent to looking for small Euclidean distance between the standardized (normalized and centered) row vectors. A correlation of $\phi$ corresponds to a distance on the sphere of $\sqrt{2 - 2\phi}$. Hence, this problem can be solved using Euclidean Locality Sensitive Hashing (LSH). Negative correlation outlier pairs can be found by simply considering every row and its negation.

Assuming outlier correlations are greater than $\phi_0$, and non-outlier correlations are smaller than $\phi_1$. We can use a Fast Johnson-Lindenstrauss Transformation to compress input rows to length $O(\epsilon^{-2}\log{n})$, distorting the pairwise distances by at most $(1 \pm \epsilon)$. Then we can do LSH with $c^2 = \left(\frac{1-\epsilon}{1+\epsilon}\right)^2 \left(\frac{1-\phi_0}{1-\phi_1}\right)$.

The best known Euclidean LSH algorithms have $\rho = \frac{1}{c^2} + o(1)$ (data independent, \cite{LSH}) and $\rho = \frac{1}{2c^2 - 1} + o(1)$ (data dependent, \cite{LSH2}). This gives us a time and space cost of $O(n^{2-\Theta(\phi_0)})$ for fixed $\phi_0$, even if $\epsilon$ goes to $1$ and $\phi_1$ goes to 0.

\noindent{\bf Compressed Matrix Multiplication.}
Pagh~\cite{OuterProducts} considered the problem of efficiently
computing sparse or approximate matrix products. The key idea is that
by choosing a particular structure for the sketching functions, it is
possible to quickly compute a sketch of the outer product
$\mathbf{xy}^T$, from sketches of $\mathbf{x}$ and $\mathbf{y}$,
through the use of FFTs (in time $O(b\log{b})$ for length $b$
sketches).
Since a matrix product $\mathbf{AB}^T$ can be decomposed into a sum of
such outer products between corresponding columns, this allows for
efficient computation of matrix products from sketches of columns.

As the algorithm only requires access to matched columns of
$\mathbf{A}$ and $\mathbf{B}$ one at a time, in the special case of
$\mathbf{A} = \mathbf{B}$ this approach can be used in the \CPS model
to build a sketch of $\mathbf{AA}^T$.
In particular, we can build a sketch of the covariance matrix $\mathbf{V}$ in this streaming model, from input observation matrix $\mathbf{M}$, with update time cost $O(b\log{b})$ ($O(1)$ amortized, since $n$ dominates $b\log{b}$) and space usage $O(b)$.
To recover dominant entries from these sketches, Pagh describes an
approach (building on Gilbert et al.~\cite{SparseRecovery}) that uses
$O(\log^2{n})$ sketches of sub-matrices of $\mathbf{AB}$, along with
error correcting codes, to discover the identity of a small number of
entries which dominate the Frobenius norm of the product, with high
probability.
This process runs in $O(b\log^2{n})$ time and space.
Putting these pieces together provides a solution to a
%Applying this to covariance sketches can solve a
{\em covariance outliers} version of our problem in the \CPS model.

Unfortunately, this approach cannot be adapted directly to the {\em
  correlation outliers} problem.
Large correlations between low variance signals would be drowned out
by the contribution from high variance signals that are much more
weakly correlated. To apply this technique, we would need to record
the whole of $\mathbf{M}$ (perhaps feasible for small $np$), or
perform two passes over the stream---using the first to determine the
variances, and then using the covariance solution on the rescaled
inputs with the second pass.
Instead, we will adapt the recovery process to work on different kinds of sketches.

\para{Boolean Vectors.}
Karppa et al.~\cite{FasterBooleanCorrelation} (improving on work by Valiant~\cite{BooleanCorrelation}) considered the problem of identifying a small number of highly correlated pairs of \emph {Boolean} vectors (entries are $\pm1$) from a collection of vectors which mostly have low pairwise correlation. Several tricks are used to beat LSH for small $\phi$:

\begin{itemize}
\item
  Calculate tensor powers of each signal vector to amplify the relative difference between the small and large correlations.
\item
  Sub-sample from the tensor power vectors to reduce dependence on the blown up dimension of these vectors. This is possible because Boolean vectors have a uniform distribution of weight among their entries.
\item
  Sum together groups of the sub-sampled vectors (pre-multiplied by random signs) to reduce dependence on the number of vectors. This is possible because only a few of the pairwise correlations are large.
\item
  Use fast matrix multiplication to quickly compute the inner products between pairs of groups. These are then inspected for unusually large entries, which indicates the presence of a large pair in the group.
\end{itemize}

The first two points depend strongly on the vectors being Boolean and
hence do \emph{not} apply in our more general real-valued setting.
Note however, that our approach shares the use of fast matrix
multiplication and the sketch-like idea of signed group
aggregation for increasing speed.

\subsection{Our Contributions}\label{sec:contributions}
We describe an algorithm which answers the correlation outliers
problem in the turnstile streaming model.
We analyze its space and time costs, and show that they meet the
desiderata above.
Our algorithm stores a separate Fast AMS sketch of each row of
$\mathbf{M}$ (described in Section~\ref{sec:sketch}).
Comparing these directly would still take time $\Theta(n^2)$ to
perform an all-pairs comparison.
Instead, we achieve an improved query time with the following three ideas:
\begin{itemize}
\item
  By randomly assigning variables into $\Pi$ groups, and linearly
  combining the (sketched) information of all variables in the same
  group we can go from having to consider $n^2$ pairs of variables to $\Pi^2$ pairs of groups. This can be seen as a second level of sketching.
\item
  Error correcting codes are composed with the grouping step (including/excluding variables from groups based on code bits) in a way that allows us to recover the identities of large entries in a particular group pair using the decoder.
\item
  Fast matrix multiplication algorithms allow us to quickly generate batches of sketch estimates of inner products. This speeds up the evaluation of the inner products between pairs of groups. Then checking whether the results of these computations exceeds a given threshold produces the strings of bits for the decoder.
\end{itemize}

Our main result is stated in full in Theorem~\ref{thm:main}.
As an example, setting an internal parameter $\theta = 2/3$ and for fixed $k, R$,
we obtain subquadratic space and time
costs summarized in the below table.

\begin{center}
  \begin{tabular}{| c | c | c | c | c |}
    \hline
    \textbf{Technique} & \textbf{Models} & \textbf{Sketch Size} & \textbf{Query Space} & \textbf{Query Time}\\
    \hline
    Fast AMS Sketches & All & $\tilde{O}(\phi^{-2}n)$ & $\tilde{O}(\phi^{-2}n)$ & $\tilde{O}(\phi^{-2}n^2)$\\
    FJLT + LSH & All & $\tilde{O}(n)$ & $\tilde{O}(n^{2-\Theta(\phi)})$ & $\tilde{O}(n^{2-\Theta(\phi)})$\\
    Our Approach & All & $\tilde{O}(\phi^{-2}n^{5/3})$ & $\tilde{O}(\phi^{-2}n^{5/3})$ & $\tilde{O}(\phi^{-2}n^{5/3})$\\
    \hline
  \end{tabular}
\end{center}

This space usage is $o(np)$ for $p \in \Omega(n^{2/3 + \epsilon})$.

\subsection{Sketches of Vectors}
\label{sec:sketch}
Our results make use of \emph{sketches} of vectors. These can be thought of as random projections from the original high-dimensional space down to a lower dimensional space, such that geometric properties of the vectors are (approximately) preserved. In particular, given vectors $\mathbf{x}$ and $\mathbf{y}$, sketches exist that can estimate:
%
%\begin{enumerate}
%\item
\begin{align}
&\text{
  Squared Euclidean length $\lVert \mathbf{x} \rVert_2^2$ up to error
  $\epsilon \lVert \mathbf{x} \rVert_2^2$.}\label{eq:euclid} \\
%\item
&\text{
  Inner product $\langle \mathbf{x} , \mathbf{y} \rangle$ up to error $\epsilon \lVert \mathbf{x} \rVert_2 \lVert \mathbf{y} \rVert_2$.}
\label{eq:inner}
\end{align}
%\end{enumerate}

Many results for such sketches are known, from the earliest (non-constructive) results based on the Johnson-Lindenstrauss lemma \cite{Johnson:Lindenstrauss:84}, the tug-of-war sketches due to Alon, Matthias, Szegedy and Gibbons~\cite{Alon:Matias:Szegedy:96,Alon:Gibbons:Matias:Szegedy:99}, and several more~\cite{Achlioptas:01,Li:Hastie:Church:07,Kane:Nelson:14}. For concreteness, we will adopt the so-called (fast) AMS sketches (explained in~\cite{Cormode:11}). These create a sketch of size $O(\epsilon^{-2} \log{1/\delta} )$ so that any query obtains the above claimed $\epsilon$ guarantee with probability at least $1-\delta$, where the probability is over the random choices used to determine the random projection.

The AMS sketching procedure maps (linearly and randomly) the space of $p$-dimensional vectors to the space of $d \times b$ matrices. Each row of the output sketch is obtained by pre-multiplying the input vector by a diagonal matrix whose entries are Rademacher (uniformly random $\pm 1$), and then pre-multiplying by a $b \times p$ sparse matrix where each column has a single $1$, with $0$ everywhere else\footnote{The construction does not require the entries to be chosen fully independently at random, so it is common to describe the sketch transformations in terms of hash functions drawn from limited independence families. This allows the transform to be stored in polylogarithmic space.}. This process generates one row of the sketch, and is repeated independently $d$ times to generate all rows. The stated $(\epsilon$, $\delta)$ guarantee can be achieved for $d \in \Theta(\log{1/\delta})$ and $b \in \Theta(\epsilon^{-2})$. As they are a sparse linear transformation of their input, any addition to an entry in the sketched vector can be applied to the sketch in time $O(d) = O(\log 1/\delta)$.

\begin{definition}\label{def:ams}
Let:
\begin{itemize}
\item
  $\textsc{AMS}_{\epsilon, \delta}$ refer to a distribution of random
  linear maps corresponding to fast AMS sketches with the stated
  $(\epsilon, \delta)$ norm and inner product approximation guarantees
  (\eqref{eq:euclid} and \eqref{eq:inner}).
\item
  \emph{$(\epsilon, \delta)$-sketch transformation} $\mathbf{S}$ be a linear map drawn from $\text{AMS}_{\epsilon, \delta}$.
\item
  The symbol $\odot$ represent the binary operation of performing the inner product query between two sketches. So
  $\mathbf{S}(\mathbf{x}) \odot \mathbf{S}(\mathbf{y}) \approx \langle \mathbf{x}, \mathbf{y} \rangle \text{, and }\, \mathbf{S}(\mathbf{x}) \odot \mathbf{S}(\mathbf{x}) \approx \|\mathbf{x}\|_2^2 \text{.}$
\end{itemize}
\end{definition}

One application of sketches is to estimate the value of a particular index in a vector. This can be achieved as a special case of an inner product query: we use the sketch to estimate $\langle \mathbf{x} , \mathbf{e}_i \rangle$, where $\mathbf{e}_i$ is the vector that is $1$ at location $i$ and $0$ elsewhere. The guarantee ensures that we obtain an estimate with error at most $\epsilon \lVert \mathbf{x} \rVert_2$. This use of sketches is referred to as a Count sketch~\cite{Charikar:Chen:Farach-Colton:02}.

\section{Algorithm and Analysis}\label{sec:theory}
\subsection{Algorithm Overview}\label{sec:overview}
Our algorithm works in the most general stream model we considered, the turnstile model. At a high level, our algorithm consists of:

\begin{itemize}
\item
  An \emph{initialization} procedure to set up the sketch data structure.
\item
  An \emph{update} procedure to process updates from the stream.
\item
  A \emph{query} procedure to recover the suspected elements of $\textsc{Large}_{\phi, k}$.
\end{itemize}

Our sketch structure is built on top of a collection of AMS sketches with standard initialization and update procedures, plus some additional variables to keep running totals. We will briefly review these procedures in Section \ref{sec:rowsketch}, as well as discussing some basic properties and routines required for the query algorithm.

The query process itself is based on two main ideas. First, we can
take linear combinations of AMS sketches and then perform inner
product queries between them in order to estimate certain kinds of
linear combinations of entries of $\mathbf{C}$. Further, we can
perform batches of such queries quickly using fast matrix
multiplication. And secondly, we can utilize error correcting codes to
identify the large magnitude entries in these linear combinations of
entries
%(as Pagh did, see Section \ref{sec:related})
even with the error introduce by the AMS sketches.

Rather than fast matrix multiplication between combinations of rows, we could have hoped to employ Pagh's compressed matrix multiplication for our second layer of sketching. However, to produce a $w$-bucket sketch would take time $\frac{w\log{w}}{\epsilon^2}$ for each row of the AMS sketches. Then to control the variance of the output buckets, we would need $w = n^2\epsilon^2$, giving a time cost of $\Omega(n^2)$ to build the secondary sketch.

The outline and the discussion of the query algorithm is therefore
broken up into four parts. In Section \ref{sec:cartesian} we describe
a ``Cartesian sketch'' which compresses a matrix by applying a pair of
independent transformations (each akin to the Count sketch), one row-wise and one column-wise. Then, in Section \ref{sec:recovery} we show that we can use the error correcting code technique to recover large entries from Cartesian sketches. Further, we show that the recovery technique is robust to additional sources of noise per entry of the sketch. Next, in Section \ref{sec:approximation} we show how the AMS sketches in our structure can be used to build good enough approximations of the Cartesian sketches to satisfy the noise limits. Finally, in Section \ref{sec:analysis} we analyze the overall space and time costs and discuss how to amplify the probability of success.
For brevity, full proofs are deferred to the Appendix, and we present
informal proofs in the main body to convey the high level ideas. 

\subsection{Row Sketching}\label{sec:rowsketch}
For our data structure we will keep an AMS sketch of each row of the observation matrix along with a running total. The choice of sketch parameters $(\epsilon, \delta)$ will be made in the final analysis in Section~\ref{sec:analysis}.

To initialize the structures, we randomly pick an $(\epsilon,
\delta)$-sketch transformation $S$, and initialize
$n$ sketches $\mathbf{r^{(i)}}$ to all zeros.
We also create $n$ counters $t^{(i)}$, initialized to zero.
%Algorithm \ref{alg:initialization} describes the set up process,
%while
Algorithm \ref{alg:update} shows how to apply a received update in the
\TS model.
We use $\mathbf{e_j}$ to indicate the length $p$-vector consisting of
a $1$ in entry $j$ and $0$ everywhere else.
The update simply updates the $i$th sketch with index $j$, and updates
the corresponding sum of weights, $t^{(i)}$.
Let $\mathbf{y^{(i)}}$ to refer to the $i$\textsuperscript{th} row of $\mathbf{M}$ for $i \in [n]$. By following these procedures we will have $\mathbf{r^{(i)}} = \mathbf{S}(\mathbf{y^{(i)}})$ and $t^{(i)} = \sum_{j \in [p]} \mathbf{y^{(i)}}_j$ at the conclusion of the stream.

%\begin{algorithm}
%  \DontPrintSemicolon
% % \Begin{
%    Randomly pick $(\epsilon, \delta)$-sketch transformation $\mathbf{S}$\;
%    \For{$i \in [n]$}{
%      Initialize $\mathbf{r^{(i)}}$ as zero sketch $\mathbf{S}(\mathbf{0}) = \mathbf{0}$\;
%      Initialize $t^{(i)}$ as $0$\;
%%    }
%  }
%  \caption{\textsc{Initialization}\label{alg:initialization}}
%\end{algorithm}

\begin{figure}[t]
  \begin{minipage}{2.75in}
   \begin{algorithm}[H]
  \DontPrintSemicolon
  \KwIn{TS model update $u_s = (\alpha, i, j)$}
  \BlankLine
%  \Begin{
    $\mathbf{r^{(i)}} \gets \mathbf{r^{(i)}} + \alpha \cdot \mathbf{S}(\mathbf{e_j})$\;
    $t^{(i)} \gets t^{(i)} + \alpha$\;
%  }
  \caption{\textsc{Update}\label{alg:update}}
  \end{algorithm}
  \end{minipage}
  \hfill
  \begin{minipage}{2.5in}
\begin{algorithm}[H]
  \DontPrintSemicolon
%%  \Begin{
    \For{$i \in [n]$}{
      $\mathbf{r^{(i)}} \gets \mathbf{r^{(i)}} - (t^{(i)}/p) \cdot \mathbf{S}(\mathbf{e})$\;
%      Compute $z^{(i)}$, result of query $$\;
      $\mathbf{r^{(i)}} \gets (\mathbf{r^{(i)}} \odot \mathbf{r^{(i)}})^{-1/2} \cdot \mathbf{r^{(i)}}$\;
    }
%%  }
  \caption{\textsc{Standardize}\label{alg:standardize}}
\end{algorithm}
  \end{minipage}
\end{figure}

An important operation we will need to be able to perform on these row sketches is to \emph{standardize} them. Recalling $\mathbf{\overline{x}}$ and $\mathbf{e}$ from Definition~\ref{def:estimators}, we define:
\begin{definition}
  For a given row vector $\mathbf{y^{(i)}}$, the \emph{standardized vector} $\mathbf{\hat{y}^{(i)}}$ is given by:
  \[ \mathbf{\hat{y}^{(i)}} = (\mathbf{y^{(i)}} -
  \mathbf{\overline{x}}_i \mathbf{e}^T) \big/ \|\mathbf{y^{(i)}} - \mathbf{\overline{x}}_i \mathbf{e}^T\|_2 \text{.} \]
  If we have a sketch $\mathbf{S}(\mathbf{y^{(i)}})$ we will refer to $\mathbf{S}(\mathbf{\hat{y}^{(i)}})$ as the \emph{standardized sketch}.
\end{definition}

In the \RPS and \CPS models we could keep track of the running sums of $\alpha^2$ for each row, allowing us to compute the exact rescaling factor required to standardize the sketches. However, in the more general \TS setting, the best we can do is an approximation.
Algorithm~\ref{alg:standardize} describes the procedure for computing the approximately standardized sketches.

Initially, it may appear that to perform this standardization at query
time, we need to spend $\Omega(pd)$ time building the sketch
$\mathbf{S}(\mathbf{e})$.
However, we can amortize this cost during the update phase.
As long as at least $p$ entries of the final $\mathbf{M}$ are
non-zero, then we can build up $\mathbf{S}(\mathbf{e})$ one entry per
update by using a single counter to track which entries have been
added.
In the atypical case that $\mathbf{M}$ is extremely sparse, we will need to add $O(pd)$ to the query time to complete the construction of this sketch.

\begin{proposition}\label{prop:standardized}
After performing the \textsc{standardize} routine, the inner product query between sketches $\mathbf{r^{(i)}}$ and $\mathbf{r^{(j)}}$ produces an estimate of $\mathbf{C}_{i,j}$ having $4\epsilon$ additive error with probability at least $1 - 3\delta$, for $\epsilon < 1/2$.
\end{proposition}
\begin{informal}
Each sketch approximates the sketch of a standardized row
($\mathbf{r^{(i)}} \approx \mathbf{S}(\mathbf{\hat{y}^{(i)}}$) and the
inner product query between sketches of standardized rows approximates
the correlation ($\mathbf{S}(\mathbf{\hat{y}^{(i)}}) \odot
\mathbf{S}(\mathbf{\hat{y}^{(j)}}) \approx \mathbf{C}_{i,j}$). To get
a small additive error on our estimates, we then just need both
sketches to be approximated well and the inner product query between
them to give a good results. Each of the three events occurs with probability $(1-\delta)$.

The correlation between two rows can be expressed as the inner product
of the corresponding standardized vectors. The sketches output by \textsc{standardize} approximate the true standardized sketches. To get a small additive error on our estimate, we then just need both sketches to be approximated well and the inner product query between them to give a good result. Each of the three occurs with probability $(1 - \delta)$.
\end{informal}

\subsection{Cartesian Sketches} \label{sec:cartesian}
\begin{definition}\label{def:cartesian}
  For an $n \times n$ matrix $\mathbf{A}$, we call $\textsc{Cart}(\mathbf{A})$ a $\Pi \times \Pi$ Cartesian sketch of $\mathbf{A}$ if for each $(h, g) \in [\Pi]^2$ we have
  \[\textsc{Cart}(\mathbf{A})_{h,g} = \sum\limits_{P_1(x)=h}\sum\limits_{P_2(y)=g} \left( s_1(x) s_2(y) \mathbf{A}_{x, y} \right) \text{,}\]
  where $s_1$ and $s_2$ are independently selected from a pairwise independent family of random sign functions $[n] \to \lbrace -1, +1 \rbrace$, and where $P_1$ and $P_2$ are functions $[n] \to [\Pi]$ selected independently and uniformly at random from the set of functions:
  \[\lbrace f : [n] \to [\Pi] \text{ s.t. } |f^{-1}(i)| = n/\Pi \text{ for each } i \in [\Pi] \rbrace \text{.}\]
\end{definition}

From this definition, we can see that a Cartesian sketch transformation is very similar to a pair of independent Count sketch transformations (one performed row-wise, one column-wise). The  difference is the use of fully random partitioning functions which produce exactly equal buckets. If we were performing exactly a pair of Count sketches, we would also have $P_1$ and $P_2$ expressed as limited independence hash functions. However, the $O(n\log{n})$ space needed to store fully random permutations will not impact our asymptotic space usage and makes the subsequent analysis simpler.
The entries of the sketches will be referred to as buckets, and the $(i,j)$\textsuperscript{th} entry of the original matrix $\mathbf{A}$ is said to be mapped to the $(h,g)$\textsuperscript{th} bucket (for a given choice of sketch functions) if $P_1(i) = h$ and $P_2(j) = g$.

\begin{definition}\label{def:bucket}
 Let $\mathcal{B}_{h,g} = \lbrace (i, j) \in [n]^2 \text{ s.t. } P_1(i) = h$  and  $P_2(j) = g \rbrace \text{,}$
  i.e. the set of index pairs mapped to bucket $(h, g)$.
\end{definition}

\subsection{Recovery Process}\label{sec:recovery}
Now we will describe how to apply the recovery process to a series of
Cartesian sketches of a given matrix. We will describe an algorithm
which gives a constant probability of finding any given element of
$\textsc{Large}_{\phi, k}$ (Definition~\ref{def:large}) and argue that it works.

For this procedure, we apply an error correcting code to encode row
and column indices (which take values in $[n]$) into a longer binary
codeword.
We will assume access to some family of functions (over choices of
$n$) with the desired properties to perform the encoding and
decoding.
For a fixed $n$, let:
\[\mathcal{E}: [n] \to \lbrace 0, 1 \rbrace^{L\log{n}}\, \text{ and }\, \mathcal{D} : \lbrace 0, 1 \rbrace^{L \log{n}} \to [n]\]

\noindent
where $L > 1$ indicates how much bigger the codeword is compared to the
input size.
Here, we write $\mathcal{E}$ and $\mathcal{D}$ for the encoder and
decoder functions (respectively) of a scheme which can recover from up
to $\lambda L \log{n}$ bit flip errors --- i.e. an error rate of
$\lambda$.
That is, for any length $L\log{n}$ binary word $\mathbf{w}$ with at most $\lambda L \log{n}$ bits set to $1$, we have\footnote{Here $\oplus$ represents the ``exclusive-or'' bitwise operation between binary words.} $\mathcal{D}(\mathcal{E}(i) \oplus \mathbf{w}) = i$ for every $i \in [n]$.

Error correcting codes are known to exist for $L \in O(1)$ and $\lambda \in \Omega(1)$, which can be implemented to perform encoding and decoding in $O(\log{n})$ time and $O(\text{polylog}\,n)$ space (for example \cite{FastECC}).

\begin{definition}\label{def:maskingmatrix}
For each $l \in [L\log{n}]$ (each bit in the code words), we define a
\emph{masking matrix} $\mathbf{E^{(l)}}$. This is a diagonal binary
matrix where entry $\mathbf{E^{(l)}}_{i,i}$ is the
$l$\textsuperscript{th} bit of the code word $\mathcal{E}(i)$.
That is, $\mathbf{E^{(l)}}_{i,i} = \mathcal{E}(i)_l$.
\end{definition}

These masking matrices can be pre- or post-multiplied with $\mathbf{C}$ to mask rows or columns (respectively) based on bits of their index encodings.

The recovery process is described in Algorithm \ref{alg:simplerecovery}. It takes as input sketches $\mathbf{L^{(l)}} = \textsc{Cart}(\mathbf{CE^{(l)}})$ and $\mathbf{R^{(l)}} = \textsc{Cart}(\mathbf{E^{(l)}C})$ for each $l \in [L\log{n}]$, for a randomly selected Cartesian sketch transformation $\textsc{Cart}$.

\begin{algorithm}[t]
  \caption{\textsc{RecoveryStep}\label{alg:simplerecovery}}
  \KwIn{Cartesian sketches $\mathbf{L^{(l)}}$ and $\mathbf{R^{(l)}}$ for $l \in [L\log{n}]$, and Cartesian sketch transformation $\textsc{Cart}$}
  \KwOut{Index multiset $\Omega$ of suspected large entries}
  \DontPrintSemicolon
  \BlankLine
%  \Begin{
    Create new empty multiset $\Omega$\;
    \For{$(h,g) \in [\Pi]^2$}{
      Create new empty strings $\mathcal{I}$ and $\mathcal{J}$\;
      \For{$l \in [L\log{n}]$}{
        \lIf{$|\mathbf{L^{(l)}} - \textsc{Cart}(\mathbf{E^{(l)}})| \geq \phi/2$}{
          Append $1$ to $\mathcal{I}$
          \textbf{else}
          Append $0$ to $\mathcal{I}$
          }
        \lIf{$|\mathbf{R^{(l)}} - \textsc{Cart}(\mathbf{E^{(l)}})| \geq \phi/2$}{
          Append $1$ to $\mathcal{J}$
          \textbf{else}          Append $0$ to $\mathcal{J}$
        }
      }
      Append $(\mathcal{D}(\mathcal{I}), \mathcal{D}(\mathcal{J}))$ to $\Omega$\;
    }
    \Return{$\Omega$}
%  }
\end{algorithm}

To understand why this process should work, consider the special case
where $\mathbf{C}$ is $0$ on all the non-large, off-diagonal
entries. That is, the only non-zero entries are the diagonals (which
must be $1$) and the entries corresponding to elements of
$\textsc{Large}_{\phi, k}$. In this situation, the only entries
contributing to the Cartesian sketches are entries of
$\textsc{Large}_{\phi, k}$ corresponding to unmasked rows and
columns. Now, consider what happens in a bucket with a single large
entry mapped to it. Whenever the row or column of the large entry is
masked, the corresponding bucket value will be $0$; and when the row
and column are not masked, the bucket value will have magnitude at
least $\phi$ --- in particular, greater than $\phi/2$. This means
that, in the algorithm, on the outer loop corresponding to this
bucket, $\mathcal{I}$ and $\mathcal{J}$ will be exactly the code words
corresponding to the row and column indices of the large entry. Hence,
the index pair of the large entry is added to $\Omega$. 
So, isolated large entries will be correctly recovered in this special case, and as long as $k$ is sufficiently smaller than $\Pi$ we have a good chance of any given large entry being isolated.
To formalize this argument, and extend it to the more general case, we define a few different events.

\begin{definition} \label{def:events} For fixed $\mathbf{C}$ and a fixed coding scheme define the following random events, over the random choice of sketch functions:
  \begin{itemize}
    \item Let $\textsc{CorrectDecode}_{h,g}$ be the event that the recovery process returns the ``correct'' index for bucket $(h,g)$. If there is exactly one of $\textsc{Large}_{\phi, k}$ in the bucket, then the correct result is the index pair of that entry. Otherwise, any returned value is considered correct.

    \item Let $\textsc{SmallError}_{h,g,l}$ be the event that the non-large entries of $\mathbf{CE^{(l)}} - \mathbf{E^{(l)}}$ and $\mathbf{E^{(l)}C} - \mathbf{E^{(l)}}$ each contribute less than $\phi/4$ to bucket $(h,g)$ of their corresponding sketches. That is, $\textsc{Cart}(\mathbf{E^{(l)}C_{-k}})_{h,g}$ has magnitude smaller than $\phi/4$.
  \end{itemize}
\end{definition}

We begin with a proposition explaining the circumstances we are looking for to successfully find a large entry.

\begin{proposition} \label{prop:fulldecode} If we have that:
%  \begin{equation}
  $
  \Prob{\textsc{CorrectDecode}_{h,g}} \geq 1-x \text{
      for all } (h,g) \in [\Pi]^2 \text{,}
  %    \label{eq:fulldecode}
  $
%    \end{equation}
  then for any given $(i, j) \in \textsc{Large}_{\phi, k}$, we have that $(i, j)$ is in the list of index pairs produced by \textsc{RecoveryStep} with probability at least $1 - x - 2k/\Pi$.
\end{proposition}

From this proposition, we can see that if we can get a lower bound on
the probability of $\textsc{CorrectDecode}_{h,g}$ for every $(h,g) \in
[\Pi]^2$, then we can get an overall guarantee for the recovery process.

\begin{proposition} \label{prop:correctdecode} If we have that:
  $
    \Prob{\textsc{SmallError}_{h,g,l}} \geq 1 - y \text{
      for all } l \in [L\log{n}] \text{,}$
%\label{eq:smallerr}
%  \end{equation}
  then $\Prob{\textsc{CorrectDecode}_{h,g}} \geq 1 - y/\lambda$.
\end{proposition}

The last piece we need is a lower bound on the probability of $\textsc{SmallError}_{h,g,l}$.

\begin{proposition} \label{prop:smallerror}
  Recalling Definition~\ref{def:large}, we have:
  $\Prob{\textsc{SmallError}_{h,g,l}} \geq (1 - 32 \|\mathbf{C_{-k}}\|_F^2 / (\Pi^2 \phi^2)) \text{.}$
\end{proposition}

Now we have all the pieces we need to show that the recovery process works.

\begin{proposition} \label{prop:recoverycorrect}
  If we have that
  $\Pi \geq \max\lbrace 18k, 18 \|\mathbf{C_{-k}}\|_F / (\phi \lambda^{1/2}) \rbrace \text{,}$
  then the output of $\textsc{RecoveryStep}$ will include any fixed index pair in $\textsc{Large}_{\phi, k}$ with probability at least $\frac23$.
\end{proposition}

Observe that we chose the definition of the event $\textsc{SmallError}_{h,g,l}$ to leave room for an additional source of noise of similar size $\phi/4$. We will need this robustness later.

\begin{corollary} \label{prop:recoveryrobust}
  Lemma \ref{prop:recoverycorrect} holds even when there is additional noise applied to each bucket entry, provided it has magnitude smaller than $\phi/4$ with probability at least $(1 - \lambda/18)$ on any fixed bucket.
\end{corollary}

In the next subsection, we will show that an approximate Cartesian sketch can be constructed from row sketches within these tolerances.

\subsection{Approximation from Row Sketches}\label{sec:approximation}
We need a way of quickly approximating $\textsc{Cart}(\mathbf{E^{(l)}C})$ and $\textsc{Cart}(\mathbf{CE^{(l)}})$ for each $l \in [L\log{n}]$ for a randomly chosen Cartesian sketch transformation $\textsc{Cart}$, from the row sketches described in Section \ref{sec:rowsketch}. This is done by the procedure described in Algorithm \ref{alg:approximation}.

\begin{algorithm*}[t]
  \caption{\textsc{Approximate}\label{alg:approximation}}
  \KwIn{Approximately standardized row sketches $\mathbf{r^{(1)}},\, \cdots,\, \mathbf{r^{(n)}}$ and functions $P_1$, $P_2$, $s_1$, $s_2$ corresponding to a Cartesian sketch transformation $\textsc{Cart}$}
  \KwOut{$\mathbf{L^{(l)}}$ and $\mathbf{R^{(l)}}$, estimates of $\textsc{Cart}(\mathbf{E^{(l)}C})$ and $\textsc{Cart}(\mathbf{CE^{(l)}})$ respectively, for $l \in [L\log{n}]$}
  \DontPrintSemicolon
  \BlankLine
%  \Begin{
    \For{$l \in [L\log{n}]$}{
      Initialize empty matrices $\mathbf{L^{(l)}}$ and $\mathbf{R^{(l)}}$\;
      \For{$h \in [\Pi]$}{
        Initialize $\textsc{left}[h]$, $\textsc{right}[h]$, $\textsc{leftMasked}[h]$, $\textsc{rightMasked}[h]$ as zero sketches $\mathbf{S}(\mathbf{0}) = \mathbf{0}$\;
      }
      \For{$i \in [n]$}{
        $\textsc{left}[P_1(i)] \gets \textsc{left}[P_1(i)] + s_1(i) \cdot \mathbf{r^{(i)}}$\;
        $\textsc{leftMasked}[P_1(i)] \gets \textsc{leftMasked}[P_1(i)] + \mathbf{E^{(l)}}_{i,i} \cdot s_1(i) \cdot \mathbf{r^{(i)}}$\;
        $\textsc{right}[P_2(i)] \gets \textsc{right}[P_2(i)] + s_2(i) \cdot \mathbf{r^{(i)}}$\;
        $\textsc{rightMasked}[P_2(i)] \gets \textsc{rightMasked}[P_2(i)] + \mathbf{E^{(l)}}_{i,i} \cdot s_2(i) \cdot \mathbf{r^{(i)}}$\;
      }
      \For{$(h, g) \in [\Pi]^2$}{
        $\mathbf{L^{(l)}}_{h,g} \gets \textsc{leftMasked}[h] \odot \textsc{right}[g]$\;
        $\mathbf{R^{(l)}}_{h,g} \gets \textsc{left}[h] \odot \textsc{rightMasked}[g]$\;
      }
    }
    \Return{$\mathbf{L^{(1)}},\, \cdots,\, \mathbf{L^{(L\log{n})}}$ and $\mathbf{R^{(1)}},\, \cdots,\, \mathbf{R^{(L\log{n})}}$}
%  }
\end{algorithm*}

For each $l \in [L\log{n}]$, the returned $\mathbf{L^{(l)}}$ is our approximate $\textsc{Cart}(\mathbf{E^{(l)}C})$ and $\mathbf{R^{(l)}}$ is our approximate $\textsc{Cart}(\mathbf{CE^{(l)}})$.

The algorithm works by observing that $\mathbf{C}$ can be approximated from the row sketches by performing all the possible inner product queries between pairs of sketches and placing the results in the corresponding positions of the matrix. We could then apply \textsc{Cart} to the result.
However, we make the further observation that since \textsc{Cart} can be broken up into pieces that look like pre- and post-multiplication by matrices, we can rearrange the order of operation. We can perform the \textsc{Cart} sketch first, directly on the row sketches, and then perform the all-pairs inner product query second.
This simple change results in the main performance bottle-neck (the all-pairs inner product query) happening on a much smaller matrix, greatly speeding up the entire query process.

We will show that this process produces a good enough approximation of the \textsc{Cart} sketch to act as the input to \textsc{RecoveryStep}.

\begin{proposition}\label{prop:approximationbound}
  At the end of \textsc{Approximate}, for any given $(h,g) \in [\Pi]^2$, we have:
  \begin{multline*}
    |\mathbf{L^{(l)}}_{h,g} - \textsc{Cart}(\mathbf{E^{(l)}C})_{h,g}| \leq \epsilon n \Pi^{-1} 207 \lambda^{-1/2} \text{,}\\
    \text{and } |\mathbf{R^{(l)}}_{h,g} - \textsc{Cart}(\mathbf{CE^{(l)}})_{h,g}| \leq \epsilon n \Pi^{-1} 207 \lambda^{-1/2} \text{,}
  \end{multline*}
  with probability at least $1 - \lambda/27 - \delta(2 + 12n/\Pi)$, as long as $\epsilon < 1/2$.
\end{proposition}

To meet the requirements for \textsc{Recovery} to work on these approximations, we need to set limits on the choices of $\epsilon$ and $\delta$.

\begin{proposition}\label{prop:approximation}
  If we have that:
 $\delta \leq \lambda/(54(2 + 12n/\Pi))$, and
$\epsilon \leq \min\lbrace 1/2, (\phi \Pi \lambda^{1/2}) / (828n)\rbrace \text{,}$
  then \textsc{Approximate} produces approximations which are within the noise tolerance of \textsc{RecoveryStep}.
\end{proposition}

\subsection{Analysis of Algorithm}\label{sec:analysis}
Putting together the previous subsections, we can make the full recovery algorithm. The outline is listed in Algorithm \ref{alg:full}.

\begin{algorithm*}[t]
  \caption{\textsc{Recover}\label{alg:full}}
  \KwIn{Row sketches $\mathbf{r^{(1)}},\, \cdots,\, \mathbf{r^{(n)}}$ and totals $t^{(1)},\, \cdots,\, t^{(n)}$}
  \KwOut{Index set $\Omega$ of entries we are confident are large}
  \DontPrintSemicolon
  \BlankLine
%  \Begin{
    Create empty multiset $\Omega$\;
    \For{$\gamma \in [\Gamma]$}{
      Randomly generate functions $P_1$, $P_2$, $s_1$, $s_2$ for a Cartesian sketch $\textsc{Cart}$\;
      Run \textsc{Approximate}, passing it $P_1$, $P_2$, $s_1$, $s_2$ and the row sketches\;
      Run \textsc{RecoveryStep}, passing it \textsc{Cart} and the result of \textsc{Approximate}\;
      Append the result of \textsc{RecoveryStep} to $\Omega$\;
    }
    Remove entries from $\Omega$ appearing fewer than $\Gamma/2$ times\;
%    Remove duplicates in $\Omega$\;
    \Return{$\Omega$ (as a set)}
 % }
\end{algorithm*}

\begin{proposition}\label{prop:final}
If we have that:
\[\delta \leq \lambda/(54(2 + 12n/\Pi))\text{, }
%\]
%\[\
%\quad
\epsilon \leq \min\lbrace 1/2, (\phi \Pi \lambda^{1/2}) / (828n)\rbrace
\text{,}
%\]
\text{ and }
%     [\
       \Pi \geq \max\lbrace 18k, 18 \|\mathbf{C_{-k}}\|_F / (\phi \lambda^{1/2}) \rbrace \text{,}\]
then we can choose a $\Gamma \in O(\log{n})$ such that \textsc{Recover} returns every element of $\textsc{Large}_{\phi, k}$ with probability at least $1 - n^{-3}$.
\end{proposition}

\begin{proposition}\label{prop:spacetime}
\textsc{Recover} can be implemented to run in time $\tilde{O}(\Gamma (\Pi^2 + \log{(1/\delta)} (n\epsilon^{-2} + \mathcal{M}(\Pi, \epsilon^{-2}))))$, and space $\tilde{O}(\Pi^2 + n \epsilon^{-2} \log{(1/\delta)})$, where $\tilde{O}(\cdot)$ is the $O(\cdot)$ cost with $\log{n}$ factors suppressed.
\end{proposition}

\begin{theorem}\label{thm:main}
  For every $\theta \in [0, 1]$, there exists a sketch of size
  $\displaystyle \tilde{O}\left( n^{2\theta}(k^2+\frac{R^2}{\phi^2}) + n^{3-2\theta}(\frac{\phi^2}{k^2\phi^2 + R^2}) \right)$
from which we can extract the (up to $k$) entries with magnitude at
least $\phi$ in time \newline
$\displaystyle \tilde{O}\left (n^{2\theta}(k^2+\frac{R^2}{\phi^2}) + n^{3-2\theta}(\frac{\phi^2}{k^2\phi^2 + R^2}) + \mathcal{M} \left( n^\theta(k+\frac{R}{\phi}), n^{2-2\theta}(\frac{\phi^2}{k^2\phi^2 + R^2} \right) \right)$
with high probability.
\end{theorem}

\begin{corollary}\label{prop:example}
In particular, for $\theta = 2/3$ and constant $k$, $R$, we can build a sketch of size $\tilde{O}\left( \phi^{-2}n^{5/3} \right)$ with query time $\tilde{O}\left( \phi^{-2}n^{5/3} \right)$.
\end{corollary}

\section{Concluding Remarks}\label{sec:conclusions}
We have shown how to guarantee accurate recovery of
correlation outliers using a sketch-based method, beating LSH on query time for small correlation outliers and vanishing correlation non-outliers.
A key part of our approach is to use sketching and coding ideas
repeatedly: as well as using sketching to reduce the initial
dimensionality of the data, we use a second ``layer'' of sketching
when we combine subsets of signals, in order to speed up queries over
many pairs of sketches at the cost of increased error. Where LSH tries to hash the correlated elements together, we try to separate them and then recover them from the noise.
This produces a trade-off between the size of the underlying sketches and the final query time. This general approach could work in other situations where a large number of sub-queries need to be evaluated to search for large values, for example with measures of similarity/distance other than correlation.

Further, as the technique produces a linear intermediate sketch, this approach is easily adapted to recover pairs whose correlation deivates from some expected correlation matrix, or has changed comapred with some previous point in time (simply perform the heavy hitters recover on the difference between two intermediate sketches built using the same permuations, signs, and codes).

Future directions would include finding ways to use alternative primitives to fast matrix multiplication (such as fast convolution via FFT, as adopted by Pagh) and trying to combine the advantages of LSH-based methods and heavy-hitters-based methods. 

%
%In this paper, we used the standard notion of fast matrix
%multiplication to quickly build the intermediate sketches in
%asymptotically fast time.
%However, in practice the best known fast matrix multiplication
%algorithms are not preferable on realistic input sizes.
%the best known theoretical algorithms for fast matrix
%multiplication are far from practical, and require truly immense input
%sizes to demonstrate the true superiority.
%Instead, we would expect that a moden fast matrix multiplication
%algorithm that is {\em not} optimal would be the appropriate choice
%for an implementation.
%A direction for further thought would be the suitability of other fast
%algorithmic constructions.
%For example, in Pagh's work, FFT is used to combine sketches.

%For the specific problem of correlation outliers, it also seems that
%LSH methods should be able to be modified to work in this more general setting. For %example, if LSH could be applied on top of the fast AMS sketch, this could provide a method which is less sensitive to the weight of the non-large entries in the correlation matrix.

\subsection*{Acknowledgements}
We thank Milan Vojnovic for several discussions about this work.

\appendix
\section{Detailed Proofs}
\begin{proof}[Proof of Lemma \ref{prop:standardized}]
Recalling definition \ref{def:estimators}, $\mathbf{C}_{i,j}$ can be expressed as $\mathbf{V}_{i,i}^{-1/2} \mathbf{V}_{i,j} \mathbf{V}_{j,j}^{-1/2}$, where each $\mathbf{V}_{h,g}$ is the scaled inner product between standardized rows $\mathbf{\hat{y}^{(h)}}$ and $\mathbf{\hat{y}^{(g)}}$:
\[\mathbf{V}_{h,g} = \frac{1}{p-1} (\mathbf{y^{(h)}} - \mathbf{\overline{x}_he}^T)(\mathbf{y^{(g)}} - \mathbf{\overline{x}_ge}^T)^T \text{.}\]

Observe that factors in $V_{h,g}$ involving $p-1$
cancel in the expression for $\mathbf{C}_{i,j}$, so they can be
ignored.
What remains is the inner product between normalized (to Euclidean norm $1$) versions of vectors $\mathbf{y^{(i)}} - \mathbf{\overline{x}_ie}^T$ and $\mathbf{y^{(j)}} - \mathbf{\overline{x}_je}^T$.

Before performing \textsc{standardize}, we had each $\mathbf{r^{(i)}} = \mathbf{S}(\mathbf{y^{(i)}})$. We also have that $(t^{(i)} / p) = \mathbf{\overline{x}}_i$. This means that at the end of the routine, each $\mathbf{r^{(i)}}$ is now a sketch of $(z^{(i)})^{-1/2} (\mathbf{y^{(i)}} - \mathbf{\overline{x}_ie}^T)$, where $(z^{(i)})^{-1/2}$ is the correct normalization factor to within multiplicative error in the range $[(1 - 2\epsilon)^{1/2}, (1 + 2\epsilon)^{1/2}]$ with probability at least $1 - \delta$.

For the result, we require two such rescaling factors to be within
their bounds, and also for the inner product query to succeed.
Each of these three events holds with probability at least $1-\delta$,
giving an overall probability at least $1 - 3\delta$ by the union bound.

To determine the overall error, consider that since $\epsilon < 1/2$
and $|\mathbf{C}_{i,j}| \leq 1$, \newline
%\begin{align*}
\centerline{$  (\mathbf{\tilde{y}^{(h)}} (1 \pm 2\epsilon)^{1/2}) \odot
(\mathbf{\tilde{y}^{(g)}} (1 \pm 2\epsilon)^{1/2})
%&
\in (1 \pm 2\epsilon) \mathbf{C}_{i,j} + \epsilon(1 \pm 2\epsilon)
%\\&
\subset \mathbf{C}_{i,j} \pm 4\epsilon \text{.}$}
%\end{align*}
\end{proof}

\begin{proof}[Proof of Lemma \ref{prop:fulldecode}]
Let $(h,g) = (P_1(i), P_2(j))$ be the bucket $(i,j)$ is mapped to. Since the partition functions are chosen uniformly at random, the chance that none of the other entries mapped to the same bucket are in $\textsc{Large}_{\phi, k}$ is at least $1 - 2k/\Pi$. To see this, observe that in the worst case, all index pairs in $\textsc{Large}_{\phi, k}$ have either the same row or same column index. Then, by the Markov inequality, we have less than $2k/\Pi$ probability that at least one of the remaining $k-1$ entries in $\textsc{Large}_{\phi, k}$ are mapped into one of the remaining $n/\Pi - 1$ slots in that bucket.

Now, if entry $(i,j)$ turns out to be the only large entry in its bucket, then the event $\textsc{CorrectDecode}_{h,g}$ occurring implies that the index pair recovered from bucket $(h,g)$ will be $(i,j)$. The chance of both occurring is then at least $1 - x - 2k/\Pi$.
\end{proof}

\begin{proof}[Proof of Lemma \ref{prop:correctdecode}]
  In the event that bucket $(h,g)$ contains more or less than one large entry, then $\textsc{CorrectDecode}_{h,g}$ automatically holds, so we only need to consider the case of exactly one large entry in the bucket.

  Now, consider the case of only one large entry $\mathbf{C}_{i,j}$
  being mapped to the bucket. Observe that we can write
  \[\textsc{Cart}(\mathbf{E^{(l)}C} - \mathbf{E^{(l)}})_{h,g} = \textsc{Big} + \textsc{Small}\text{,}\]
  where $\textsc{Big} = \mathbf{E^{(l)}}_{i,i} \cdot s_1(i) \cdot s_2(j) \cdot \mathbf{C}_{i,j}$ and $\textsc{Small} = \textsc{Cart}(\mathbf{E^{(l)}C_{-k}})_{h,g}$ (see Definition~\ref{def:large}).

  When the event $\textsc{SmallError}_{h,g,l}$ holds, we have that $|\textsc{Small}| \leq \phi/4$. Also, $|\textsc{Big}|$ is either $0$ (when the row of the large entry is masked) or greater than $\phi$ (when not masked). So, $\textsc{SmallError}_{h,g,l}$ holding means that the $l$\textsuperscript{th} threshold bit will match the $l$\textsuperscript{th} bit of the code word for the row index we are trying to recover. An analogous argument applies to $\textsc{Cart}(\mathbf{CE^{(l)}} - \mathbf{E^{(l)}})_{h,g}$ and the column index.

  For the decoder to correctly recover an index from its code word, we need at most a $\lambda$ fraction of errors. So, we need less than a $\lambda$ fraction of the events $\textsc{SmallError}_{h,g,l}$ for $l \in [L\log{n}]$ failing to hold. By Markov's inequality, we can put the chance of more than a $\lambda$ fraction of failures at less than $y/\lambda$.
\end{proof}

\begin{proof}[Proof of Lemma \ref{prop:smallerror}]
  For fixed $(h,g,l) \in [\Pi]^2 \times [L\log{n}]$, consider the random variable $\textsc{Cart}(\mathbf{E^{(l)} C_{-k}})_{h,g}$ (random over the choices of $P_1$, $P_2$, $s_1$, and $s_2$ that make up $\textsc{Cart}$). This can be broken down into a sum of contributions from each entry of $\mathbf{E^{(l)} C_{-k}}$, as follows:
  \[ \textsc{Cart}(\mathbf{E^{(l)} C_{-k}})_{h,g} = \sum\limits_{(i,j) \in [n]^2} \beta_{i,j,l}
  \]
  \[ \text{ where }  \beta_{i,j,l} =
  \begin{cases}
    (\mathbf{E^{(l)}}_{i,i}) s_1(i) s_2(j) (\mathbf{C_{-k}})_{i,j} & \text{if } (i, j) \in \mathcal{B}_{h,g}\\
    0 & \text{otherwise,}\\
  \end{cases}
  \]
  recalling from definition~\ref{def:bucket} that $\mathcal{B}_{h,g}$ represents the index pairs mapped to bucket $(h, g)$.

  Due to the independently selected pairwise independent random sign functions $s_1$ and $s_2$, each term has mean $\E{\beta_{i,j,l}} = 0$ and covariance $\text{Cov}[\beta_{i_1,j_1,l}, \beta_{i_2,j_2,l}] = 0$ (where either $i_1 \neq i_2$ or $j_1 \neq j_2$). This means the variance of the sum (the bucket value) is simply the sum of the variances of the terms.

  Each term has variance $\Var{\beta_{i,j,l}} \leq (\mathbf{E^{(l)} C_{-k}})_{i,j}^2 / \Pi^2$. To see this, observe that each term has at most a $1/\Pi^2$ chance of being non-zero (due to the random partition functions). Summing up all the terms gives us \begin{align*}
    \Var{\textsc{Cart}(\mathbf{E^{(l)} C_{-k}})_{p,q}} &=
    \|\mathbf{E^{(l)} C_{-k}}\|_F^2 / \Pi^2
%    \\&
    \leq \|\mathbf{C_{-k}}\|_F^2 / \Pi^2\text{.}
  \end{align*}
  Then, by Chebyshev's inequality, we have
  $\textsc{Cart}(\mathbf{E^{(l)}C_{-k}})_{p,q} \geq \phi/4 \text{,}$
  with probability less than $16 \|\mathbf{C_{-k}}\|_F^2 / (\Pi^2 \phi^2)$.
  An analogous argument applies to $\textsc{Cart}(\mathbf{C_{-k} E^{(l)}})_{h,g}$, giving the result by union bound.
\end{proof}

\begin{proof}[Proof of Lemma \ref{prop:recoverycorrect}]
  Substituting $\Pi \geq 18 \|\mathbf{C_{-k}}\|_F / (\phi \lambda^{1/2})$ into Lemma~\ref{prop:smallerror} gives us that:
  \[\Prob{\textsc{SmallError}_{h,g,l}} \geq 1 - 8\lambda/81 \text{.}\]

  Then by Lemma~\ref{prop:fulldecode} (with $y = 8\lambda/81$), we get that:
  \[\Prob{\textsc{CorrectDecode}_{h,g}} \geq 1 - 8/81 \text{.}\]

  Finally, using the fact that $\Pi \geq 18k$ (from the initial assumptions) along with Lemma~\ref{prop:correctdecode} (with $x = 8/81$), we have that any fixed $(i, j) \in \textsc{Large}_{\phi, k}$ will be in the output of \textsc{RecoveryStep} with probability at least $1 - 8/81 - 1/9 = 64/81 \geq 2/3$.
\end{proof}

\begin{proof}[Proof of Corollary \ref{prop:recoveryrobust}]
  Observe that the proof of Lemma~\ref{prop:correctdecode} still works with an additional term of magnitude no more than $\phi/4$. Then, observe that the choice of parameters in Lemma \ref{prop:recoverycorrect} leaves enough slack to condition on an additional event occuring with probability greater than $1 - \lambda/18$ per $\textsc{SmallError}_{h,g,l}$.
\end{proof}

\begin{proof}[Proof of Lemma \ref{prop:approximationbound}]
  If we performed the algorithm with the exact vectors instead of AMS sketches, then $\mathbf{L^{(l)}}_{h,g}$ would be exactly $\textsc{Cart}(\mathbf{E^{(l)}C})$. Any difference is due to the inner product approximation error which is smaller than $\epsilon \lVert \mathbf{h} \rVert_2 \lVert \mathbf{g} \rVert_2$ with probability at least $1 - \delta$, where $\mathbf{h}$ and $\mathbf{g}$ are the vectors that $\textsc{LeftMasked}[h]$ and $\textsc{Right}[g]$ are sketches of.
  First consider
  \[\mathbf{h} = \sum_{P_1(i) = h} (\mathbf{E^{(l)}}_{i,i} \cdot s_1(i) \cdot \mathcal{R}^{(i)} \cdot \mathbf{\tilde{y}^{(i)}}) \text{,}\]
  where $\mathcal{R}^{(i)}$ is the rescaling error caused by \textsc{Standardize} (see Section~\ref{sec:rowsketch}). Recall that each $|\mathcal{R}^{(i)}| \leq 1 + 4\epsilon$ with probability at least $1 - 3\delta$ as long as $\epsilon < 1/2$.

  The squared $2$-norm $\lVert \mathbf{h} \rVert_2^2$ is given by:
  \[ \sum\limits_{P_1(i) = P_1(j) = h} \mathcal{R}^{(i)} \mathcal{R}^{(j)} \langle (\mathbf{E^{(l)}}_{i,i} s_1(i) \mathbf{\tilde{y}^{(i)}}), (\mathbf{E^{(l)}}_{j,j} s_1(j) \mathbf{\tilde{y}^{(j)}}) \rangle
  \]
  For each $i = j$, the corresponding term is equal to $\mathcal{R}^{(i)}\mathcal{R}^{(j)}$, and for each $i < j$ there is a matching equal term with $i$ and $j$ swapped. So, with probability at least $1 - 3n\delta/\Pi$ the norm is at most $n (1 + 4\epsilon)^2/\Pi$ plus an independent random variable (random over choice of $s_1$) with mean $0$ and variance less than
  \[4\lVert \mathbf{C} \rVert_F^2 (1 + 4\epsilon)^2 / \Pi^2 \leq 4n^2 (1 + 4\epsilon)^2 / \Pi^2 \text{.}\]
  So by Chebyshev's inequality and a union bound, $\lVert\mathbf{h}\rVert_2^2$ is smaller than $\frac{n}{\Pi}(1 + 4\epsilon)^2(22\lambda^{-1/2} + 1)$ with probability greater than $1 - \lambda/108 - 3n\delta/\Pi$. The same bound applies to $\lVert \mathbf{g} \rVert_2^2$, so $\epsilon \lVert \mathbf{h} \rVert_2 \lVert \mathbf{g} \rVert_2 \leq \epsilon (1 + 4\epsilon)^2 n \Pi^{-1} 23\lambda^{-1/2} \leq \epsilon n \Pi^{-1} 207 \lambda^{-1/2}$ with probability at least $1 - \lambda/54 - \delta(1 + 6n/\Pi)$.

  An analogous argument works for entry $\mathbf{R^{(l)}}_{h,g}$ and $\textsc{Cart}(\mathbf{CE^{(l)}})$. A union bound over the probabilities of failure gives the result.
\end{proof}

\begin{proof}[Proof of Lemma \ref{prop:approximation}]
  The assumptions imply that:
  \[\epsilon n \Pi^{-1} 23\lambda^{-1/2} \leq 23\phi/828 \leq \phi/4
  \text{, and }
  1 - \delta(2 + 12n/\Pi) - \lambda/27 \leq 1 - \lambda/18 \text{.}\]

  This tells us exactly that the errors on the approximations according to Lemma~\ref{prop:approximationbound} are within the bounds allowed by Lemma~\ref{prop:recoveryrobust}.
\end{proof}

\begin{proof}[Proof of Lemma \ref{prop:final}]
From Lemmas \ref{prop:recoverycorrect} and \ref{prop:approximation} we know for $\Gamma = 1$, this algorithm succeeds at finding any one large entry with probability at least $2/3$. By performing $O(\log{n})$ independent repetitions and then only considering those index pairs appearing at least half the time, then by the Chernoff bound we can amplify the probability of finding any one of the large entries to $1 - n^{-5}$. There are at most $n^2$ such pairs, giving the result.
\end{proof}

\begin{proof}[Proof of Lemma \ref{prop:spacetime}]
\textsc{RecoveryStep} can be implemented to run in time $O(\Pi^2 \text{polylog} n)$ since we have $\Pi^2$ iterations of the outer loop, $O(\log{n})$ iterations of the inner loop, and all operations taking $O(\text{polylog} n)$ time (coding schemes with such fast decoding algorithms exist).

\textsc{Approximate} can be implemented to run in time $O(\log{n} \log{(1/\delta)}( n \epsilon^{-2} + \mathcal{M}(\Pi, \epsilon^{-2})))$ where $\mathcal{M}(\Pi, \epsilon^{-2})$ is the time required to multiply a $\Pi \time \epsilon^{-2}$ matrix by an $\epsilon^{-2} \time \Pi$ matrix. This holds because there are $O(\log{n})$ iterations of the outer loop. Then within we have $O(n)$ additions involving sketches of size $O(\epsilon^{-2} \log{(1/\delta)})$. We also have a series of inner products which can be expressed as a batched all-pair query. This can be performed as a series of $O(\log{(1/\delta)})$ matrix multiplications.

Putting this together, we get a time cost of $\tilde{O}(\Gamma (\Pi^2 + \log{(1/\delta)} (n\epsilon^{-2} + \mathcal{M}(\Pi, \epsilon^{-2}))))$. The filtering step adds no extra asymptotic time, since we can filter by sorting the $O(\Pi^2 \log{n})$ pairs and then iterating over them counting repetitions to see if any exceed the $\Gamma/2$ threshold.

\textsc{RecoveryStep} uses $O(\Pi^2 \log{n} + \text{polylog} n)$ space to store a pair of length $O(\log{n})$ strings, a multiset of up to $\Pi^2$ index pairs, and the input of $O(\log{n})$ $\Pi$-by-$\Pi$ sketches. The polylog overhead is used for the encoding scheme.

\textsc{Approximate} uses $O(\Pi^2 \log{n} + n \epsilon^{-2} \log{(1/\delta)})$ space to store $O(n)$ sketches and $O(\log{n})$ $\Pi \times \Pi$ matrices.

All together we need $\tilde{O}(\Pi^2 + n \epsilon^{-2} \log{(1/\delta)})$ space, since the multiset contains at most $O(\Pi^2 \log{n})$ index pairs.
\end{proof}

\begin{proof}[Proof of Theorem \ref{thm:main}]
  Substitute bounds in Lemma~\ref{prop:final} into costs in Lemma~\ref{prop:spacetime}.
\end{proof}

\newpage
\bibliography{sources}

\end{document}